\documentclass{llncs}

\usepackage{latexsym,ifthen,epsfig,color}
\usepackage{float}
\usepackage{amsmath,amssymb,algorithm2e,enumerate,float}
\usepackage{pstricks, pst-node, multido}

\newcommand{\NP}{\ensuremath{\mathbb{NP}}}

\newtheorem{algo}{Algorithm}

\title{Efficient Dominating and Edge Dominating Sets for Graphs and Hypergraphs}

\author{
Andreas Brandst\"adt
\inst{1}
\and
Arne Leitert
\inst{1}
\and
Dieter Rautenbach
\inst{2}
}

\institute{
Institut f\"ur Informatik, Universit\"at Rostock, Germany.\\
\email{ab@informatik.uni-rostock.de, arne.leitert@uni-rostock.de}\medskip
\and
Institut f\"ur Optimierung und Operations Research, Universit\"at Ulm, Germany\\
\email{dieter.rautenbach@uni-ulm.de}
}

\begin{document}

\maketitle

\begin{abstract}
Let $G=(V,E)$ be a graph. A vertex {\em dominates} itself and all its neighbors, i.e., every vertex $v \in V$ dominates its closed neighborhood $N[v]$. 
A vertex set $D$ in $G$ is an {\em efficient dominating} ({\em e.d.}) set for $G$ if for every vertex $v \in V$, there is exactly one $d \in D$ dominating $v$. An edge set $M \subseteq E$ is an {\em efficient edge dominating} ({\em e.e.d.}) set for $G$ if it is an efficient dominating set in the line graph $L(G)$ of $G$. The ED problem (EED problem, respectively)  asks for the existence of an e.d. set (e.e.d. set, respectively) in the given graph.  

We give a unified framework for investigating the complexity of these problems on various classes of graphs. In particular, we solve some open problems and give linear time algorithms for ED and EED on dually chordal graphs. 

%It follows from our unified framework and known results that both problems are solvable in polynomial time on AT-free graphs which also solves an open problem for EED on permutation graphs.  

We extend the two problems to hypergraphs and show that ED remains \NP-complete on $\alpha$-acyclic hypergraphs, and is solvable in polynomial time on hypertrees, while EED is polynomial on $\alpha$-acyclic hypergraphs and \NP-complete on hypertrees. 
\end{abstract}


\section{Introduction and Basic Notions}\label{sec:intro}

Packing and covering problems in graphs and their relationships belong to the most fundamental topics in combinatorics and graph algorithms and have a wide spectrum of applications in computer science, operations research and many other fields. Recently, there has been an increasing interest in problems combining packing and covering properties. Among them, there are the following variants of domination problems: 

Let $G$ be a finite simple undirected graph with vertex set $V$ and edge set $E$. 
%A vertex set $U \subseteq V$ is {\em independent} in $G$ if for all $x,y \in U$, $xy \not\in E$ holds. 
A vertex {\em dominates} itself and all its neighbors, i.e., every vertex $v \in V$ dominates its closed neighborhood $N[v]=\{u \mid u=v$ or $uv \in E\}$. 
A vertex set $D$ in $G$ is an {\em efficient dominating} ({\em e.d.}) set for $G$ if for every vertex $v \in V$, there is exactly one $d \in D$ dominating $v$ 
(sometimes called {\em independent perfect dominating set}) \cite{BanBarSla1988,BanBarHosSla1996}.
An edge set $M \subseteq E$ is an {\em efficient edge dominating} ({\em e.e.d.}) set for $G$ if it is an efficient dominating set in the line graph $L(G)$ of 
$G$ (sometimes called {\em dominating induced matching}).
The ED problem (EED problem, respectively) asks for the existence of an e.d. set (e.e.d. set, respectively) in the given graph. Note that both problems are \NP-complete. The complexity of ED (EED, respectively) (and their variants) with respect to special graph classes was studied in various papers; see e.g. \cite{BanBarHosSla1996,Biggs1973,ChaPanCoo1995,LiaLuTan1997,Lin1998,LuTan2002,Milan2011,Yen1992,YenLee1996} for ED and  
\cite{BraHunNev2010,BraMos2011,CarKorLoz2011,GriSlaSheHol1993,LuKoTan2002,LuTan1998} for EED.
The main contributions of our paper are:  
\begin{enumerate}[$(i)$]
\item a unified framework for the ED and EED problems solving some open questions,
\item linear time algorithms for ED and EED on dually chordal graphs, and
\item an extension of the two problems to hypergraphs, in particular to $\alpha$-acyclic hypergraphs and hypertrees: We show that ED remains \NP-complete on $\alpha$-acyclic hypergraphs, and is solvable in polynomial time on hypertrees, while EED is polynomial on $\alpha$-acyclic hypergraphs and \NP-complete on hypertrees. 
\end{enumerate}
Our approach has the advantage that it unifies the proofs of various results obtained in numerous papers. Our proofs are typically very short since we extensively use some theoretical background on the relations of the considered graph and hypergraph classes, and in particular closure properties of graph classes with respect to squares of their graphs, and polynomial time algorithms for Maximum Weight Independent Set and 
Minimum Weight Dominating Set on some graph classes. 
The consequences are some new cases where the corresponding problems can be efficiently solved. 
 
%Due to space limitation, all proofs are given in an Appendix.

\section{Further Basic Notions}\label{sec:basicnotions}

\subsection{Basic notions and properties of graphs}

Let $G$ be a finite undirected graph without loops and multiple edges. Let $V$ denote its vertex (or node) set and $E$ its edge set; let $|V|=n$ and $|E|=m$. A vertex $v$ is \emph{universal} in $G$ if it is adjacent to all other vertices of $G$.
A {\em chordless path} $P_k$ ({\em chordless cycle} $C_k$, respectively) has $k$ vertices, say $v_1,\ldots,v_k$, and edges $v_iv_{i+1}$, $1 \le i \le k-1$ (and $v_kv_1$, respectively). A {\em hole} is a chordless cycle $C_k$ for $k \ge 5$. $G$ is {\em chordal} if no induced subgraph of $G$ is isomorphic to $C_k$, $k \ge 4$. See e.g. \cite{BraLeSpi1999} for the many facets of chordal graphs. 
A vertex set $U \subseteq V$ is {\em independent} if for all $x,y \in U$, $xy \notin E$ holds. 
For a graph $G$ and a vertex weight function on $G$, let the {\sc Maximum Weight Independent Set} (MWIS) problem be the task of finding an independent vertex set of maximum weight.

Let $K_i$ denote the clique with $i$ vertices. 
Let $K_4-e$ or {\em diamond} be the graph with four vertices and five edges, say vertices $a,b,c,d$ and edges $ab,ac,bc,bd,cd$; its {\em mid-edge} is the edge $bc$.
Let {\em gem} denote the graph consisting of five vertices, four of which induce a $P_4$, and the fifth is universal.  
Let $W_4$ denote the graph with five vertices consisting of a $C_4$ and a universal vertex. 

For $U \subseteq V$, let $G[U]$ denote the subgraph induced by $U$.
For a set ${\cal F}$ of graphs, a graph $G$ is called {\em ${\cal F}$-free} if $G$ contains no induced subgraph from ${\cal F}$. Thus, it is {\em hole-free} if it contains no induced subgraph isomorphic to a hole. A graph $G$ is {\em weakly chordal} if $G$ and its complement graph are hole-free.  
Three pairwise non-adjacent distinct vertices form an {\em asteroidal triple} ({\em AT}) in $G$ if for each choice of two of them, there is a path between the two avoiding the neighborhood of the third. A graph $G$ is {\em AT-free} if $G$ contains no AT. 

\subsection{Basic notions and properties of hypergraphs}
Throughout this paper, a hypergraph $H=(V,{\cal E})$ has a finite vertex set $V$ and for all $e \in {\cal E}$, $e \subseteq V$ (${\cal E}$ possibly being a multiset).
For a graph $G$, let ${\cal N}(G)$ denote the {\em closed neighborhood hypergraph}, i.e., ${\cal N}(G)=(V,\{N[v] \mid v \in V\})$, and let ${\cal C}(G)$ denote the {\em clique hypergraph} consisting of the inclusion-maximal cliques of $G$. 

A subset of edges ${\cal E}' \subseteq {\cal E}$ is an {\em exact cover} of $H$ if for all $e,f \in {\cal E}'$ with $e \neq f$, $e \cap f = \emptyset$ and $\bigcup {\cal E}'=V$. 
The {\sc Exact Cover} problem asks for the existence of an exact cover in a given hypergraph $H$. It is well known that this problem is \NP-complete even for 3-elementary hyperedges (problem X3C [SP2] in \cite{GarJoh1979}). Thus, the ED problem on a graph $G$ is the same as the Exact Cover problem on  ${\cal N}(G)$. 

For defining the class of dually chordal graphs, whose properties will be contrasted with those of chordal graphs, as well as for extending the ED and the EED problems to hypergraphs, we need some basic definitions: 
For a hypergraph $H=(V,{\cal E})$, let $2sec(H)$ denote its {\em $2$-section} (also called {\em representative} or {\em primal}) graph, i.e., the graph having the same vertex set $V$ in which two vertices are adjacent if they are in a common hyperedge.  
Let $L(H)$ denote the line graph of $H$, i.e., the graph with the hyperedges ${\cal E}$ as its vertex set in which two hyperedges are adjacent in $L(H)$ if they intersect each other. 

A hypergraph $H=(V,{\cal E})$ has the {\em Helly property} if the total intersection of every pairwise intersecting family of hyperedges of ${\cal E}$ is nonempty.  
$H$ is {\em conformal} if every clique of the 2-section graph $2sec(H)$ is contained in a hyperedge of ${\cal E}$   
(see e.g. \cite{Berge1973,Fagin1983}).   
    
The notion of {\em $\alpha$-acyclicity} \cite{Fagin1983} is one of the most important and most frequently studied hypergraph notions. Among the many equivalent conditions describing $\alpha$-acyclic hypergraphs, we take the following: 
For a hypergraph $H=(V,{\cal E})$, a tree $T$ with node set ${\cal E}$ and edge set $E_T$ is a {\em join tree} of $H$ if for all vertices $v \in V$, 
the set of hyperedges containing $v$ induces a subtree of $T$. $H$ is {\em $\alpha$-acyclic} if it has a join tree. 
For a hypergraph $H=(V,{\cal E})$ and vertex $v \in V$, let ${\cal E}_v:=\{e \in {\cal E} \mid v \in e\}$. 
Let $H^*:=({\cal E}, \{{\cal E}_v \mid v \in V\})$ be the {\em dual hypergraph} of $H$.   
$H=(V,{\cal E})$ is a {\em hypertree} if there is a tree $T$ with vertex set $V$ such that for all $e \in {\cal E}$, $T[e]$ is connected. %\cite{BraDraCheVol1998,DraPriChe1992}. 

\begin{theorem}[Duchet, Flament, Slater, see \cite{BraLeSpi1999}]\label{hyptreechar}
$H$ is a hypertree if and only if $H$ has the Helly property and its line graph $L(H)$ is chordal.
\end{theorem} 

The following facts are well known: 

\begin{lemma}\label{basicdualityfacts}
Let $H$ be a hypergraph. 
\begin{enumerate}[$(i)$]
    \item $H$ is conformal if and only if $H^*$ has the Helly property. 
    \item $L(H)$ is isomorphic to $2sec(H^*)$.
\end{enumerate}
\end{lemma}

Thus:

\begin{corollary}\label{alphaacychypgrchar}
$H$ is $\alpha$-acyclic if and only if  $H$ is conformal and its $2$-section graph is chordal.
\end{corollary} 
 
It is easy to see that the dual ${\cal N}(G)^*$ of ${\cal N}(G)$ is ${\cal N}(G)$ itself, and
for any graph $G$: 
\begin{equation}\label{G2LNG}
G^2 \mbox { is isomorphic to } L({\cal N}(G)).  
\end{equation}

In \cite{BraDraCheVol1998}, the notion of dually chordal graphs was introduced: For a graph $G=(V,E)$ and a vertex $v \in V$, a vertex $u \in N[v]$ is a {\em maximum neighbor of $v$} if for all $w \in N[v]$, $N[w] \subseteq N[u]$ holds. (Note that by this definition, a vertex can be its own maximum neighbor.) 
Let $\sigma=(v_1,\ldots,v_n)$ be a vertex ordering of $V$. Such an ordering $\sigma$ is a {\em maximum neighborhood ordering} of $G$ if for every $i \in \{1,2,\ldots,n\}$, $v_i$ has a maximum neighbor in $G_i:=G[\{v_i,\ldots,v_n\}]$. A graph is {\em dually chordal} if it has a  maximum neighborhood ordering.
The following is known:

\begin{theorem}[\cite{BraCheDra1998,BraDraCheVol1998,DraPriChe1992}]\label{mainduallychgr}
Let $G$ be a graph. Then the following are equivalent: 
\begin{enumerate}[$(i)$]
\item $G$ is a dually chordal graph. 
\item ${\cal N}(G)$ is a hypertree.
\item ${\cal C}(G)$ is a hypertree.
\item $G$ is the $2$-section graph of some hypertree.
\end{enumerate}
\end{theorem} 

Thus, Theorems \ref{hyptreechar} and \ref{mainduallychgr} together with (\ref{G2LNG}) and the duality properties in Lemma \ref{basicdualityfacts} imply:  

\begin{corollary}[\cite{BraCheDra1998,BraDraCheVol1998,DraPriChe1992}]\label{linegrofalphaacychyp}
Let $G$ be a graph and $H$ be a hypergraph. 
\begin{enumerate}[$(i)$]
\item $G$ is dually chordal if and only if $G^2$ is chordal and ${\cal N}(G)$ has the Helly property.
\item If $H$ is $\alpha$-acyclic then its line graph $L(H)$ is dually chordal.
\item If $H$ is a hypertree then its $2$-section graph $2sec(H)$ is dually chordal.
\end{enumerate}
\end{corollary} 

\section{Efficient (Edge) Domination in General}

Recall that a subset $D \subseteq V$ of vertices is an efficient dominating set if for all $v \in V$, there is exactly one $d \in D$ such that $v \in N[d]$. Also a subset $M \subseteq E$ of edges is an efficient edge dominating set in $G$ if for all $e \in E$, there is exactly one $e' \in M$ intersecting the edge $e$. 

Both definitions can be extended to hypergraphs: A subset $D \subseteq V$ is an \emph{efficient dominating set} for a hypergraph~$H$ if it is an efficient dominating set for its 2-section graph $2sec(H)$. A subset $M \subseteq {\cal E}$ of hyperedges is an \emph{efficient edge dominating set} for $H$ if for all $e \in \cal E$, there is exactly one $e' \in M$ intersecting the edge $e$.

\begin{corollary}\label{cor:EEDdually}
    A vertex set $D$ is an efficient dominating set in $H$ if and only if $D$ is an efficient edge dominating set in $H^*$.
\end{corollary}

The following approach developed in \cite{Leite2012} and \cite{Milan2011} gives a tool for showing that for various classes of graphs, the ED problem can be solved in polynomial time.    
For a graph $G=(V,E)$, we define the following vertex weight function: Let $\omega(v):= |N_G[v]|$ (i.e., $\omega(v):= deg(v)+1$), and for $D \subseteq V$, let $\omega(D):=\Sigma_{d \in D}$ $\omega(d)$. Obviously, the following holds:

\begin{proposition}\label{domweightinequ}
Let $G=(V,E)$ be a graph and $D \subseteq V$.
\begin{enumerate}[$(i)$]
\item If $D$ is a dominating vertex set in $G$ then $\omega(D) \ge |V|$. 
\item If $D$ is an independent vertex set in $G^2$ then $\omega(D) \le |V|$. 
\end{enumerate}
\end{proposition} 

\begin{lemma}\label{mainequived}
Let $G=(V,E)$ be a graph and $\omega(v):= |N[v]|$ a vertex weight function for $G$. Then the following are equivalent for any subset $D \subseteq V$:
\begin{enumerate}[$(i)$]
\item $D$ is an efficient dominating set in $G$
\item $D$ is a minimum weight dominating vertex set in $G$ with $\omega(D)=|V|$.
\item $D$ is a maximum weight independent vertex set in $G^2$ with $\omega(D)=|V|$.
\end{enumerate}
\end{lemma}

\begin{proof}
(i) $\rightarrow$ (ii): If $D$ is an efficient dominating set in $G$ then, by definition, it is a dominating vertex set where every vertex $v \in V$ is in exactly one closed neighborhood of elements of $D$. Thus, the closed neighborhoods $N[d]$, $d \in D$, give a partition of $V$, and thus, $\omega(D)=|V|$ holds. Also, by Proposition~\ref{domweightinequ} there is no $D'$ with $\omega(D') < \omega(D)$.
 
\noindent
(ii) $\rightarrow$ (iii): If $D$ is a dominating set in $G$ with $\omega(D)=|V|$ then the closed neighborhoods $N[d]$, $d \in D$, give a partition of $V$, and 
in particular, $D$ is an maximum independent vertex set in $G^2$ with $\omega(D)=|V|$.

\noindent
(iii) $\rightarrow$ (i): If $D$ is an independent vertex set in $G^2$ with $\omega(D)=|V|$ then the closed neighborhoods $N[d]$, $d \in D$, give a partition of $V$,
and thus, $D$ is an efficient dominating set in $G$. 
\qed
\end{proof}

Note that $D$ is not any independent (dominating) set, but a maximum (minimum) weight one. This implies: 

\begin{corollary}\label{maincor}
For every graph class ${\cal C}$ for which the MWIS problem is solvable in polynomial time on squares of graphs from ${\cal C}$, the ED problem for ${\cal C}$ is solvable in polynomial time. 
\end{corollary}

\begin{corollary}\label{cor:mainequived}
Let $H$ be a hypergraph, $L(H)=(V,E)$ its line graph and $\omega(v):= |N[v]|$ a vertex weight function for $L(H)$ as above. Then the following are equivalent for any subset $D \subseteq V$:
\begin{enumerate}[$(i)$]
\item $D$ is an efficient edge dominating set in $H$
\item $D$ is an efficient dominating set in $L(H)$. 
\item $D$ is a minimum weight dominating vertex set in $L(H)$ with $\omega(D)=|V|$.
\item $D$ is a maximum weight independent vertex set in $L(H)^2$ with $\omega(D)=|V|$.
\end{enumerate}
\end{corollary}

\section{Efficient Domination in Graphs}

This section presents results for the ED problem on some graph classes.

\begin{theorem}[\cite{LuTan2002,YenLee1996}]\label{theo:EDchordalNPc}
    The ED problem is \NP-complete for bipartite graphs, for chordal graphs as well as for chordal bipartite graphs.  
\end{theorem}

By Corollary~\ref{linegrofalphaacychyp}~(i), the square of a dually chordal graph is chordal. Thus, based on Lemma~\ref{mainequived}, ED for dually chordal graphs can be solved in polynomial time by solving the MWIS problem on chordal graphs. However, the MWIS problem is solvable in linear time for chordal graphs with the following algorithm:

\begin{algo}[\cite{Frank1975}]\label{algo:MWISchordal}\normalfont\ \\\nopagebreak
\textbf{Input:} A chordal graph $G=(V,E)$ with $|V| = n$ and a vertex weight function~$\omega$.\\
\textbf{Output:} A maximum weight independent set $\cal I$ of $G$.

\begin{enumerate}[(1)]

    \item  Find a perfect elimination ordering $(v_1,\ldots,v_n)$ and set $\cal I := \emptyset$.
        
    \item \textbf{For} $i:= 1$ \textbf{To} $n$
    \begin{enumerate}
        \item[] If $\omega(v_i) > 0$, mark $v$ and set $\omega(u) := \max(\omega(u) - \omega(v_i), 0)$ for all vertices $u \in N(v_i)$.
    \end{enumerate}

    \item \textbf{For} $i:= n$ \textbf{DownTo} $1$
    \begin{enumerate}
        \item[] If $v_i$ is marked, set ${\cal I} := {\cal I} \cup \{v_i\}$ and unmark all $u \in N(v_i)$.
    \end{enumerate}
\end{enumerate}
\end{algo}

By using the following lemmas, the algorithm can be modified in such a way, that it solves the ED problem for dually chordal graphs in linear time.

\begin{lemma}[\cite{BraCheDra1998}]\label{lem:MNO_linear}
    A maximum neighborhood ordering of $G$ which simultaneously is a perfect elimination ordering of $G^2$ can be found in linear time.
\end{lemma}

The algorithm given in \cite{BraCheDra1998} not only finds a maximum neighborhood ordering $(v_1,\ldots,v_n)$. It also computes the maximum neighbors $m_i$ for each vertex $v_i$ with the property that for all $i < n$ no vertex $v_i$ is its own maximum neighbor ($v_i \neq m_i$). This is necessary for the following lemma.

\begin{lemma}\label{lem:ijInE_iff_mjInE2}
Let $G=(V,E)$ be a graph with $G^2=(V,E^2)$ and a maximum neighborhood ordering $(v_1, \ldots, v_n)$ where $m_i$ is the maximum neighbor of $v_i$ with $v_i \neq m_i$ and $1 \leq i < j \leq n$. Then: $v_iv_j \in E^2 \Leftrightarrow m_iv_j \in E$.
\end{lemma}

\begin{proof}
$\Leftarrow:$ $v_j$ is connected with $m_i$ ($m_iv_j \in E$). Thus, the distance between $v_i$ and $v_j$ is at most 2. So $v_i$ and $v_j$ are also connected in $G^2$ ($v_iv_j \in E^2$).

$\Rightarrow:$ $v_i$ and $v_j$ are connected in $G^2$ ($v_iv_j \in E^2$). If $v_iv_j \in E$, then $v_iv_j \in E^2$. Now assume, that $v_iv_j \notin E$. Then there is a vertex $v_k$ with $v_iv_k \in E$ and $v_kv_j \in E$. We can distinguish between two cases:
\begin{enumerate}[(i)]

    \item $i<k$. By definition $m_i$ is connected with all neighbors of $v_k$. This also includes $v_j$.
    
    \item $k<i$. In this case $v_k$ has a maximum neighbor $m_k$ with $v_im_k \in E$ and $m_kv_j \in E$. Thus, we can repeat the distinction with $v_k:=m_k$ until $i<k$.

\end{enumerate}
\qed
\end{proof}

This allows to modify Algorithm~\ref{algo:MWISchordal} in a way, that it is no longer necessary to compute the square of the given dually chordal graph $G$. Instead, a maximum weight independent set of $G^2$ can be computed on $G$ in linear time.

\begin{algo}[\label{algo:EDdc}\cite{Leite2012}]
\normalfont\ \\\nopagebreak
\textbf{Input:} A dually chordal graph $G=(V,E)$.\\
\textbf{Output:} An efficient dominating set $D$ (if existing).

\begin{enumerate}[(1)]

    \item $D = \emptyset$.
    \item \textbf{For All} $v \in V$
    \begin{enumerate}
        \item[] Set $\omega(v) := |N(v)|$ and $\omega_p(v) := 0$. $v$ is unmarked and not blocked.
    \end{enumerate}

    \item  Find a maximum neighborhood ordering $(v_1,\ldots,v_n)$ with the corresponding maximum neighbors $(m_1, \ldots, m_n)$ where $v_i \neq m_i$ for $1 \leq i < n$.
        
    \item \textbf{For} $i:= 1$ \textbf{To} $n$
    \begin{enumerate}
        \item[] For all $u \in N[v_i]$ set $\omega(v_i) := \omega(v_i) - \omega_p(u)$.
        \item[] If $\omega(v_i) > 0$, mark $v_i$ and set $\omega_p(m_i) := \omega_p(m_i) + \omega(v_i)$.
    \end{enumerate}

    \item \textbf{For} $i:= n$ \textbf{DownTo} $1$
    \begin{enumerate}
        \item[] If $v_i$ is marked and $m_i$ is not blocked, set $D := D \cup \{v_i\}$ and block all $u \in N(v_i)$.
    \end{enumerate}
    
    \item $D$ is an efficient dominating set if and only if $\sum_{v \in D}|N[v]|=|V|$.
\end{enumerate}

\end{algo}

\begin{theorem}\label{theo:EDlinearDuallyChordal}
    Algorithm~\ref{algo:EDdc} works correctly and runs in linear time.
\end{theorem}

\begin{proof}
    Algorithm~\ref{algo:EDdc} is a modification of Algorithm~\ref{algo:MWISchordal}. It computes a maximum weight independent set~$D$ of the square of the given graph~$G$ and checks if this set is an efficient dominating set.
    
    Based on Lemma~\ref{lem:MNO_linear}, line~(3) of Algorithm~\ref{algo:EDdc} computes a perfect elimination ordering of $G^2$.
    
    Let $v_i$ and $v_j$ be adjacent in $G^2$ and $i <j$. Now Lemma~\ref{lem:ijInE_iff_mjInE2} allows to modify the two loops. For the first loop (line~(4)), there is an extra vertex weight $\omega_p$. Instead of decrementing the weights $\omega$ of the neighbors of $v_i$ (and their neighbors), $\omega_p$ of the maximum neighbour~$m_i$ is incremented by $\omega(v_i)$. Now before comparing $\omega(v_j)$ to $0$, $\omega(v_j)$ is decremented by $\omega_p(u)$ for all $u \in N[v_j]$. Because $m_i$ is connected with $v_j$ (Lemma~\ref{lem:ijInE_iff_mjInE2}), this ensures that each time the weight~$\omega(v_j)$ is compared to $0$, it has the same value as it would have in Algorithm~\ref{algo:MWISchordal}.
    
    For the second loop (line~(5)) the argumentation works quite equally. After selecting a vertex~$v_j$ (i.e. $v_j \in D$), all its neighbors are blocked. Thus by Lemma~\ref{lem:ijInE_iff_mjInE2}, if a vertex~$v_i$ is adjacent in $G^2$ to a selected vertex, then the maximum neighbor $m_i$ is blocked.
    
    By Lemma~\ref{mainequived} it can be checked if $D$ is an efficient dominating set by counting the number of neighbors.
    
    \medskip
    
    Each line of Algorithm~\ref{algo:EDdc} runs in linear time. For line~(3) the algorithm given in~\cite{BraCheDra1998} can be used. The lines~(4) to~(6) are bounded by the number of nodes and their neighbors (Recall, $\sum_{v\in V}|N(v)|=2|E|$). \qed
\end{proof}

Note that strongly chordal graphs are dually chordal \cite{BraDraCheVol1998}. In \cite{LuTan2002} one of the open problems is the complexity of (weighted) ED for strongly chordal graphs which is solved by Theorem~\ref{theo:EDlinearDuallyChordal} (for the weighted case see \cite{Leite2012}).

\begin{theorem}\label{IPD-ATfrgr}
    For AT-free graphs, the ED problem is solvable in polynomial time. 
\end{theorem} 

\begin{proof}
In \cite{ChaHoKo2003}, it is shown that the square of any AT-free graph is a co-comparability graph (which is AT-free). In \cite{BroKloKraMue1999}, the MWIS problem for AT-free graphs is solved in polynomial time. This and Corollary~\ref{maincor} implies the result.
\qed
\end{proof}

This partially extends the result of \cite{ChaPanCoo1995} showing that the (weighted) ED problem for co-comparability graphs is solvable in polynomial time. 

\medskip

In \cite{LuTan2002}, one of the open problems is the complexity of ED for convex bipartite graphs. This class of graphs is contained in interval bigraphs, and a result of \cite{Keil2012} shows that the boolean width of interval bigraphs is at most $2 \log n$, based on a corresponding result for interval graphs \cite{BelVat2011}. By a result of \cite{BuiTelVat2011}, this leads to a polynomial time algorithm for Minimum Weight Domination on interval bigraphs. 

\begin{corollary}\label{cor:convexED}
For interval bigraphs, the ED problem is solvable in polynomial time.
\end{corollary} 

This solves the open question from \cite{LuTan2002} for convex bipartite graphs.

\section{Efficient Edge Domination in Graphs}

\begin{lemma}[\cite{BraHunNev2010,BraMos2011}]\label{dimlemma}
Let $G$ be a graph that has an e.e.d. set $M$.
\begin{enumerate}[$(i)$]
\item $M$ contains exactly one edge of every triangle of $G$.
\item $G$ is $K_4$-free.
\item If $xy$ is the mid-edge of an induced diamond in $G$ 
then $M$ necessarily contains $xy$. Thus, in particular, $G$ is $W_4$-free and gem-free. 
\end{enumerate}
\end{lemma}

In \cite{LuKoTan2002}, it was shown that the EED problem is solvable in linear time on chordal graphs. This allows us to solve the EED problem for dually chordal graphs using the following lemma:

\begin{lemma}\label{lem:chordIffDuallyChordal}
    Let $G$ be a graph that has an e.e.d. set. Then $G$ is chordal if and only if $G$ is dually chordal. 
\end{lemma}

\begin{proof}
    $\Rightarrow:$ $G$ is chordal. Since $G$ has an e.e.d. set there is no gem in $G$ (Lemma~\ref{dimlemma}). Thus, $G$ is sun-free chordal (strongly chordal respectively). Every strongly chordal graph is dually chordal \cite{BraDraCheVol1998}. So $G$ is dually chordal.
    
    $\Leftarrow:$ $G$ is dually chordal. Let $\sigma=(v_1,\ldots,v_n)$ be a maximum neighborhood ordering of $G$. If $G$ is not chordal then $G$ contains an induced subgraph $C$ isomorphic to $C_k$ for some $k \ge 4$. Let $w=v_i$ be the leftmost vertex of $C$ in $\sigma$. Vertex $w$ has a maximum neighbor $u$ in $G_i=G[\{v_i,\ldots,v_n\}]$. Note that $w \neq u$ since $C$ is chordless. Now, if $C$ is a $C_4$, $u$ is adjacent to all vertices of $C$, and if $C$ is a $C_k$ for $k \ge 5$ then $u$ is adjacent to all vertices of a $P_4$ in $C$ and thus contains a gem --- in each case it follows by Lemma~\ref{dimlemma} that $G$ has no e.e.d. set. 
\qed
\end{proof}

\begin{corollary}
    The EED problem can be solved in linear time for dually chordal graphs.
\end{corollary}

Efficient edge dominating sets are closely related to maximum induced matchings; it is not hard to see that every efficient edge dominating set is a maximum induced matching but of course not vice versa. However, when the graph has an efficient edge dominating set and is regular then every maximum induced matching is an efficient edge dominating set \cite{CarCerDelSil2008}. On the other hand, the complexity of the two problems differs on some classes such as claw-free graphs where the Maximum Induced Matching (MIM) problem is \NP-complete \cite{KobRot2003} (even on line graphs) while the EED problem is solvable in polynomial time \cite{CarKorLoz2011}. While every graph has a maximum induced matching, this is not the case for efficient edge dominating sets. Thus, if the graph $G$ has an efficient edge dominating set, this gives also a maximum induced matching but in the other case, the MIM problem is hard for claw-free graphs. 

For the MIM problem, there is a long list of results of the following type: If a graph $G$ is in a graph class ${\cal C}$ then also $L(G)^2$ is in the same class (see e.g. \cite{Camer2004,CamSriTan2003}), and if the MWIS problem is solvable in polynomial time for the same class, this leads to polynomial algorithms for the MIM problem on the class ${\cal C}$. A very large class of this type are interval-filament graphs \cite{Gavri2000} which include co-comparability graphs and polygon-circle graphs; the latter include circle graphs, circular-arc graphs, chordal graphs, and outerplanar graphs. AT-free graphs include co-comparability graphs, permutation graphs and trapezoid graphs (see \cite{BraLeSpi1999}). 

\begin{theorem}[\cite{Camer2004,Chang2003}]\label{LH2intervalfilATfr}
Let $G=(V,E)$ be a graph, $L(G)$ its line graph, and $L(G)^2$ its square. Then the following conditions hold:
\begin{enumerate}[$(i)$]
\item If $G$ is an interval-filament graph then $L(G)^2$ is an interval-filament graph. 
\item If $G$ is an AT-free graph then $L(G)^2$ is an AT-free graph.
\end{enumerate}
\end{theorem} 

The MWIS problem for interval-filament graphs is solvable in polynomial time \cite{Gavri2000}; as a consequence, it is mentioned in \cite{Camer2004}
that the MIM problem is efficiently solvable for interval-filament graphs. In \cite{BroKloKraMue1999} it is shown that the MWIS problem is solvable for AT-free graphs in time $O(n^4)$. For the EED problem, it follows:  

\begin{corollary}\label{intervalfilATfrdimpol}
The EED problem is solvable in polynomial time for interval-filament graphs and for AT-free graphs.
\end{corollary} 

This generalizes the corresponding result for bipartite permutation graphs in \cite{LuKoTan2002}; bipartite permutation graphs are AT-free. It also generalizes a corresponding result for MIM on trapezoid graphs \cite{GolLew2000} (which are AT-free). 

\medskip
 
In \cite{CarKorLoz2011}, the complexity of the EED problem for weakly chordal graphs was mentioned as an open problem; in \cite{BraHunNev2010}, however, it was shown that the EED problem (DIM problem, respectively) is solvable in polynomial time for weakly chordal graphs. It is easy to see that long antiholes (i.e., complements of $C_k$, $k \ge 6$) have no e.e.d. set, i.e., a hole-free graph having an e.e.d. set is weakly chordal. Thus, for hole-free graphs, the EED problem is solvable in polynomial time \cite{BraHunNev2010}. Corollary~\ref{cor:mainequived} leads to an easier way of solving the EED problem on weakly chordal graphs:

\begin{corollary}\label{cor:EEDwcg}
The EED problem is solvable in polynomial time for weakly chordal graphs.
\end{corollary}

\begin{proof}
In \cite{CamSriTan2003}, it is shown that for a weakly chordal graph $G$, $L(G)^2$ is weakly chordal. For weakly chordal graphs, MWIS is solvable in time ${\cal O}(n^4)$ \cite{HaySpiSri2000,HaySpiSri2007}. Thus, by Corollary~\ref{cor:mainequived}, the claim holds.
\qed
\end{proof}

\medskip

The next result contrasts to the fact that the MIM and the EED problem are solvable in polynomial time on chordal graphs:  

\begin{proposition}\label{MIMduallychNPc}
The MIM problem is \NP-complete for dually chordal graphs.
\end{proposition} 

\begin{proof}
Let $G$ be any graph. We add a new universal vertex, say $u$, to $G$ and denote this graph by $G'$. It is easy to see that $G'$ is dually chordal, and if $G$ is nontrivial, a maximum induced matching of $G'$ cannot contain $u$ in its vertex set. This and the fact that the MIM problem is \NP-complete in general \cite{Camer1989} shows Proposition \ref{MIMduallychNPc}. 
\qed
\end{proof}

\section{Some Results for Hypergraphs}

Theorem~\ref{theo:EDchordalNPc} and the fact that every chordal graph is the 2-section graph of an $\alpha$-acyclic hypergraph (namely, of its clique hypergraph ${\cal C}(G)$) implies: 

\begin{corollary}\label{EDalphaacyclic}
The ED problem is \NP-complete for $\alpha$-acyclic hypergraphs. 
\end{corollary} 

This situation is better for hypertrees:

\begin{corollary}\label{EDhyptree}
For hypertrees, the ED problem is solvable in polynomial time. 
\end{corollary} 

\begin{proof}
Let $H$ be a hypertree. Then by Theorem~\ref{mainduallychgr}~(iv), $2sec(H)$ is dually chordal. Thus, ED is solvable in polynomial time using Algorithm~\ref{algo:EDdc}. 
\qed
\end{proof}

Based on Corollary~\ref{cor:EEDdually} and the duality of hypertrees and $\alpha$-acyclic hypergraphs it follows:

\begin{corollary}\label{EEDhypertree}
The EED problem for hypertrees is $\NP$-complete.
\end{corollary} 

\begin{corollary}\label{EDhyptree}
For $\alpha$-acyclic hypergraphs, the EED problem is solvable in polynomial time. 
\end{corollary} 

For the MIM problem we show:

\begin{theorem}\label{MIMalphaacyclic}
    The MIM problem is solvable in polynomial time for $\alpha$-acyclic hypergraphs.
\end{theorem}

\begin{proof}
Let $H=(V,{\cal E})$ be $\alpha$-acyclic. Then by Corollary \ref{linegrofalphaacychyp} (ii), $L(H)$ is dually chordal. By Corollary \ref{linegrofalphaacychyp} (i), the square of a dually chordal graph is chordal, i.e., $L(H)^2$ is chordal.   
So the MIM problem for $\alpha$-acyclic hypergraphs can be solved by solving the MWIS problem for chordal graphs. 
\qed
\end{proof}

\smallskip

\begin{theorem}\label{MIMhypertree}
The MIM problem for hypertrees is \NP-complete. 
\end{theorem}

\begin{proof}
Let $H=(V,{\cal E})$ be a hypertree. By definition, $M$ is an induced matching for $H$ if it is an independent node set for $L(H)^2$. By Theorem~\ref{hyptreechar}, for hypertrees $H$, $L(H)$ is chordal, and by the same argument as in the proof of Theorem \ref{EEDhypertree}, for every chordal graph $G$, there is a hypertree $H$ such that $G$ is isomorphic to $L(H)$. 

Thus, the MIM problem on hypertrees corresponds to the Maximum Independent Set problem on squares of chordal graphs. This problem, however, is \NP-complete as the following reduction shows: 
Let $G=(V,E)$ be any graph, and let $F=(V \cup E \cup \{f,g \}, E_F)$ be the following graph: Its node set consists of all vertices and edges of $G$ and additionally two new nodes $f$ and $g$. Its edge set $E_F$ consists of the edges of a clique $E \cup \{f\}$ and a single edge $fg$, and all edges between any edge $e=xy \in E$ and its vertices $x$, $y$ respectively. Thus, $F$ is a split graph and therefore, $F$ is chordal. In $G'=F^2$, the vertex set $V$ of $G$ induces a subgraph isomorphic to $G$. Let $\alpha(G)$ denote the maximum size of an independent vertex set in $G$. We claim: 
\begin{equation}\label{squarechordalalpha}
\alpha(G') = \alpha(G)+1.  
\end{equation}

\noindent
{\em Proof of} (\ref{squarechordalalpha}): 
First assume that $U$ is an independent vertex set in $G$ with $|U| = \alpha(G)$. Then obviously, $U'=U \cup \{g\}$ is an independent node set in $G'$ with $|U'| = \alpha(G)+1$. Moreover, $\alpha(G') \le \alpha(G)+1$ since $E \cup \{f,g \}$ is a clique in $G'$.

Conversely assume that $U'$ is an independent node set in $G'$ with $|U'| = \alpha(G)+1$. Since $E \cup \{f,g \}$ is a clique in $G'$, $U'$ can contain at most one node of $E \cup \{f,g \}$, i.e., $|U' \cap V| = \alpha(G)$, and obviously, $U' \cap V$ is an independent vertex set in $G$ which shows the claim.  

This reduction shows Theorem \ref{MIMhypertree}.
\qed
\end{proof}

\medskip

The proof of Theorem \ref{MIMhypertree} also shows that the Maximum Independent Set problem is \NP-complete for squares of chordal and split graphs. For squares of chordal graphs this was already mentioned without proof in \cite{KanjKratsch2009}.

\begin{theorem}\label{XChypertree}
The Exact Cover problem is \NP-complete for $\alpha$-acyclic hypergraphs and solvable in polynomial time for hypertrees. 
\end{theorem}

\begin{proof}
\emph{Exact Cover for $\alpha$-acyclic hypergraphs.} Let $H=(V,{\cal E})$ be an arbitrary hypergraph. We construct $H'=(V',{\cal E}')$ as follows: $V'=V\cup \{u,v\}$ and ${\cal E}'={\cal E}\cup \{g,k\}$ with $g=V\cup \{u\}$ and $k=\{u,v\}$. $H'$ is $\alpha$-acyclic. Now we show that $H$ has an exact cover if and only if $H'$ has one.

$\Rightarrow$: Let $C \subseteq \cal E$ be an exact cover for $H$. Then $C'=C \cup \{k \}$ is an exact cover for $H'$.

$\Leftarrow$: Let $C' \subseteq \cal E'$ be an exact cover for $H'$. The vertex~$v$ is only included in $k$. Therefore $k \notin C'$, so $g$ can not be in $C'$ (otherwise $u$ is covered twice). Thus, $C= C' \setminus \{k\}$ is an exact cover for $H'$.

\medskip

\emph{Exact Cover for hypertrees.} For a hypertree $H=(V,{\cal E})$, let $L(H)=({\cal E},E)$ its line graph and $\omega$ a weight function for ${\cal E}$ with $\omega(e)=|e|$. It is easy to see, that $C\subseteq {\cal E}$ is an exact cover for $H$ if and only if $C$ is an maximum weight independent set for $L(H)$ with $\sum_{e \in C} \omega(e)=|V|$. Because $L(H)$ is chordal, the Exact Cover problem can be solved for hypertrees in polynomial time by using Algorithm~\ref{algo:MWISchordal}.
\qed
\end{proof}

Because for each hypergraph $H$ an edge set $C\subseteq {\cal E}$ is an exact cover for $H$ if and only if $C$ is an maximum weight independent set for $L(H)$, the Exact Cover problem can be solved in in polynomial time if the MWIS problem is polynomial time solvable for the line graph of $H$.

\section{Conclusion}

%In this paper we have investigated the complexity of the Efficient Domination (ED) problem and of the Efficient Edge Domination (EED) problem on various classes of graphs in a systematic way using a connection to the Maximum Weight Independent Set problem on squares of certain graphs, and we have extended these problems to hypergraphs. In particular, we focus on the duality of the Efficient Domination (ED) problem and the Efficient Edge Domination (EED) problem; while ED is \NP-complete on chordal graphs (on $\alpha$-acyclic hypergraphs, respectively) and linear on dually chordal graphs (polynomial on hypertrees, respectively), the EED problem is linear on chordal graphs (polynomial on $\alpha$-acyclic hypergraphs, respectively), and linear on dually chordal graphs while it is \NP-complete on hypertrees.  
%The closely related Maximum Induced Matching (MIM) problem is known to be polynomial on chordal graphs, and it is polynomial on $\alpha$-acyclic hypergraphs while it is \NP-complete on dually chordal graphs and on hypertrees. Both problems - ED and EED - are solvable in polynomial time on AT-free graphs. Note that the Independent Domination problem (where one looks for an independent vertex set in a given graph which is dominating in the same graph, thus also combining a packing and a covering problem) is \NP-complete for dually chordal graphs while some other variants of the domination problem are efficiently solvable on dually chordal graphs \cite{BraCheDra1998}. For chordal graphs, the independent domination problem is solvable in linear time \cite{Farbe1982}. 

The subsequent scheme summarizes some of our results; $\NP$-c. means $\NP$-complete, pol. (linear) means polynomial-time (linear-time) solvable, and XC means the Exact Cover problem. 

\begin{displaymath}
\begin{array}{|c||c|c||c|c|}
\hline
 & \mbox{ chordal gr. } & \mbox{ dually chordal gr. } & \mbox{ $\alpha$-acyclic hypergr. } & \mbox{ hypertrees }\\
\hline
\hline
\mbox{ ED } & \mbox{ $\NP$-c. \cite{YenLee1996} } & \mbox{ linear } & \mbox{ $\NP$-c. } & \mbox{ pol. }\\
\hline
\mbox{ EED } & \mbox{ linear \cite{LuKoTan2002} } & \mbox{ linear } & \mbox{ pol. } & \mbox{ $\NP$-c. }\\
\hline
\mbox{ MIM }  & \mbox{ pol. \cite{Camer1989} } & \mbox{ $\NP$-c. } & \mbox{ pol. } & \mbox{ $\NP$-c. }\\
\hline
\mbox{ XC }  & \mbox{ } & \mbox{ } & \mbox{ $\NP$-c. } & \mbox{ pol. }\\
\hline
\end{array}
\end{displaymath}

%In \cite{CarKorLoz2011}, it is mentioned that the EED problem is expressible in a certain kind of Monadic Second Order Logic, and in \cite{CouMakRot2000}, it was shown that such problems can be solved in linear time on any class of bounded clique-width assuming that the clique-width expressions are given or can be determined in the same time bound. The same remark holds for the ED and the MIM problem. Note that graphs having an e.e.d. set do not contain some forbidden subgraphs such as $K_4$ and gem. This leads to bounded clique-width in some cases.

%\medskip

\noindent
{\bf Acknowledgement.}

The first author is grateful to J. Mark Keil and Haiko M\"uller for stimulating discussions and related results.

\begin{footnotesize}

\end{footnotesize}

\begin{thebibliography}{99}

\bibitem{BanBarSla1988}
    D.W. Bange, A.E. Barkauskas, P.J. Slater,
    Efficient dominating sets in graphs,
    in: R.D. Ringeisen and F.S. Roberts, eds., Applications of Discrete Math. (SIAM, Philadelphia, 1988) 189-199.
    
\bibitem{BanBarHosSla1996}
    D.W. Bange, A.E. Barkauskas, L.H. Host, P.J. Slater,
    Generalized domination and efficient domination in graphs,
    {\sl Discrete Math.} 159 (1996) 1-11.

\bibitem{BelVat2011}
    R. Belmonte, M. Vatshelle,     
    Graph Classes with Structured Neighborhoods and Algorithmic Applications, 
    Extended abstract in: Proceedings of WG 2011, LNCS 6986, pp. 47-58, 2011.
    
\bibitem{BuiTelVat2011}
    B.-M. Bui-Xuan, J.A. Telle, M. Vatshelle,
    Boolean width of graphs,
    {\sl Theor. Computer Science} 412 (2011) 5187-5204.
    
\bibitem{Berge1973} 
    C. Berge, Graphs and Hypergraphs, North-Holland 1973 

\bibitem{Biggs1973}
   N.~Biggs, 
   Perfect codes in graphs,   
   {\sl J. of Combinatorial Theory (B)}, 15 (1973) 289-296. 

\bibitem{BraCheDra1998}
    A. Brandst\"adt,  V.D. Chepoi, and F.F. Dragan,
    The algorithmic use of hypertree structure and maximum neighbourhood orderings,
    {\sl Discrete Applied Math.} 82 (1998) 43-77.
 
\bibitem{BraDraCheVol1998}
    A. Brandst\"adt, F.F. Dragan, V.D. Chepoi, and V.I. Voloshin,
    Dually chordal graphs,
    {\sl SIAM J. Discrete Math.} 11 (1998) 437-455.
    
\bibitem{BraHunNev2010}
    A. Brandst\"adt, C. Hundt, R. Nevries,
    Efficient Edge Domination on Hole-Free graphs in Polynomial Time,
    Conference Proceedings LATIN 2010, LNCS 6034, (2010) 650-661.

\bibitem{BraLeSpi1999}
     A. Brandst\"adt, V.B. Le, and J.P. Spinrad, Graph Classes: A
     Survey, {\em SIAM Monographs on Discrete Math. Appl.}, Vol. 3, SIAM, Philadelphia, 1999.
  
\bibitem{BraMos2011}
    A. Brandst\"adt, R. Mosca,
    Dominating induced matchings for $P_7$-free graphs in linear time,
    extended abstract in: T. Asano et al. (Eds.): Proceedings of ISAAC 2011, LNCS 7074, 100-109, 2011. 
    Technical report CoRR, arXiv:1106.2772v1 [cs.DM], 2011.

\bibitem{BroKloKraMue1999} 
     H.J.~Broersma, T.~Kloks, D.~Kratsch, H.~M\"uller, 						 	 
     Independent sets in asteroidal-triple-free graphs,
     {\sl SIAM J. Discrete Math.} 12 (1999) 276-287.
     
\bibitem{Camer1989}
     K.~Cameron,
     Induced matchings,
     {\sl Discrete Applied Math.} 24 (1989) 97-102.

\bibitem{Camer2004}
    K. Cameron,
    Induced matchings in intersection graphs;
    {\sl Discrete Mathematics} 278 (2004) 1-9.

\bibitem{CamSriTan2003}
     K.~Cameron, R.~Sritharan, Y. Tang,
     Finding a maximum induced matching in weakly chordal graphs,
     {\sl Discrete Math.} 266 (2003) 133-142.

\bibitem{CarCerDelSil2008}
     D.M.~Cardoso, J.O. Cerdeira, C. Delorme, P.C. Silva, 
     Efficient edge domination in regular graphs,  
     {\sl Discrete Applied Math.} 156 (2008) 3060-3065.
     
\bibitem{CarKorLoz2011}
     D.M.~Cardoso, N. Korpelainen, V.V.~Lozin,
     On the complexity of the dominating induced matching problem in hereditary classes of graphs,
     {\sl Discrete Applied Math.} 159 (2011) 521-531.

%\bibitem{CarLoz2008}
%     D.M.~Cardozo, V.V.~Lozin,
%     Dominating induced matchings,
%``Graph Theory, Computational Intelligence and Thought'', A Conference Celebrating Marty
%Golumbic's 60th Birthday, Jerusalem, Tiberias, Haifa 2008. {\sl Lecture Notes in Computer Science} Vol. 5420 (2009) 77-86.

\bibitem{Chang2003}
    J.-M. Chang,
    Induced matchings in asteroidal-triple-free graphs,
    {\sl Discrete Applied Math.} 132 (2003) 67-78

\bibitem{ChaHoKo2003}
    J.-M. Chang, C.W. Ho, M.T. Ko,
    Powers of asteroidal-triple-free graphs with applications,
    {\sl Ars Combinatoria} 67 (2003) 161-173

\bibitem{ChaPanCoo1995}
    G.J. Chang, C. Pandu Rangan, S.R. Coorg,
    Weighted independent perfect domination on co-comparability graphs,
    {\sl Discrete Applied Math.} 63 (1995) 215-222.

%\bibitem{CouMakRot2000}
%   B. Courcelle, J.A. Makowsky and U. Rotics,
%   Linear time solvable optimization problems on graphs of bounded clique width,
%   {\sl Theory of Computing Systems} {\bf 33} (2000) 125-150.

\bibitem{DraPriChe1992}	
F.F. Dragan, C.F. Prisacaru, and V.D. Chepoi,
     Location problems in graphs and the Helly property (in Russian),
     {\sl Discrete Mathematics, Moscow}, 4(1992), 67-73 (the full version
     appeared as preprint: F.F. Dragan, C.F. Prisacaru, and V.D. Chepoi, 
     r-Domination and p-center problems on graphs: special solution methods
     and graphs for which this method is usable (in Russian), Kishinev State
     University, preprint MoldNIINTI, N. 948-M88, 1987)
     
\bibitem{Fagin1983}
    R. Fagin, Degrees of Acyclicity for Hypergraphs and Relational Database Schemes, {\sl Journal ACM}, 30 (1983) 514-550.  

%\bibitem{Farbe1982}
%    M. Farber, 
%    Independent domination in chordal graphs, 
%    {\sl Operations Research Letters} 1 (1982) 134-138.
   
\bibitem{Frank1975}
    A. Frank,
    Some polynomial algorithms for certain graphs and hypergraphs,
    Proceedings of the 5th British Combinatorial Conf. (Aberdeen 1975), 
    {\sl Congressus Numerantium} XV (1976) 211-226.

\bibitem{GarJoh1979}
    M.R. Garey, D.S. Johnson,
    Computers and Intractability -- A Guide to the Theory of NP-completeness,
    {\sl Freeman}, San Francisco, 1979.
    
\bibitem{Gavri2000}
    F. Gavril, 
    Maximum weight independent sets and cliques in intersection graphs of filaments,
    {\sl Information Processing Letters} 73 (2000) 181-188.
    	
\bibitem{GolLew2000}
    M.C. Golumbic, M. Lewenstein,
    New results on induced matchings,
    {\sl Discrete Applied Math.} 101 (2000) 157-165.
	
\bibitem{GriSlaSheHol1993}
    D.L. Grinstead, P.L. Slater, N.A. Sherwani, N.D. Holmes,
    Efficient edge domination problems in graphs,
    {\sl Information Processing Letters} 48 (1993) 221-228.

\bibitem{HaySpiSri2000}
     R.B. Hayward, J.P. Spinrad, R. Sritharan, Weakly chordal graph
     algorithms via handles, {\em Proceedings of the 11th Symposium on Discrete Algorithms}, 42-49, 2000.

\bibitem{HaySpiSri2007}
     R.B. Hayward, J.P. Spinrad, R. Sritharan, Improved algorithms for weakly chordal graphs, {\em Graphs and Combinatorics} 3(2), Art.
     14, 2007.
     
\bibitem{KanjKratsch2009}
     I. Kanj und D. Kratsch, Convex Recoloring Revisited: Complexity and Exact Algorithms, {\em Computing and Combinatorics, Lecture Notes in Computer Science} Vol. 5609 (2009) 388-397.
     
\bibitem{KobRot2003}
    D. Kobler, U. Rotics,
    Finding maximum induced matchings in subclasses of claw-free and $P_5$-free graphs,
    and in graphs with matching and induced matching of equal maximum size,
    {\sl Algorithmica} 37 (2003) 327-346

\bibitem{Keil2012}
    J.M. Keil, The dominating set problem in interval bigraphs, abstract in: Proceedings of the 3rd Annual Workshop on Algorithmic Graph Theory, Nipissing University, North Bay, Ontario, 2012
    
\bibitem{Leite2012}
   A. Leitert,
   Das Dominating Induced Matching Problem f\"ur azyklische Hypergraphen, 
   Diploma Thesis, University of Rostock, Germany, 2012.

\bibitem{LiaLuTan1997}
   Y.D. Liang, C.L. Lu, C.Y. Tang,
   Efficient domination on permutation graphs and trapezoid graphs,
   in: {\sl Proceedings COCOON'97}, T. Jiang and D.T. Lee, eds., {\sl Lecture Notes in Computer Science} Vol. 1276
   (1997) 232-241.

\bibitem{Lin1998}
   Y.-L. Lin, 
   Fast algorithms for independent domination and efficient domination in trapezoid graphs,
   in: {\sl Proceedings ISAAC'98}, {\sl Lecture Notes in Computer Science} Vol. 1533
   (1998) 267-275.

%\bibitem{LivSto1988}
%   M. Livingston, Q. Stout,
%   Distributing resources in hypercube computers,
%   in: {\sl Proceedings 3rd Conf. on Hypercube Concurrent Computers and Applications} (1988) 222-231.

\bibitem{LuKoTan2002}
    C.L. Lu, M.-T. Ko, C.Y. Tang,
    Perfect edge domination and efficient edge domination in graphs,
{\sl Discrete Applied Math.} 119 (2002) 227-250.

\bibitem{LuTan1998}
    C.L. Lu, C.Y. Tang, Solving the weighted efficient edge domination problem on bipartite permutation graphs, 
   {\sl Discrete Applied Math.} 87 (1998) 203-211.

\bibitem{LuTan2002}
    C.L. Lu, C.Y. Tang, Weighted efficient domination problem on some perfect graphs, 
   {\sl Discrete Applied Math.} 117 (2002) 163-182.

%\bibitem{McCSpi1999} 
%    R.M. McConnell, J.P. Spinrad,
%    Modular decomposition and transitive orientation,
%    {\sl Discrete Math.} 201 (1999) 189-241.

\bibitem{Milan2011}
   M. Milani\v c,  
   A hereditary view on efficient domination, 
   extended abstract in: Proceedings of the 10th Cologne-Twente workshop 2011, pp. 203-206. Full version to appear under the title ``Hereditary efficiently dominatable graphs''. 

\bibitem{Yen1992}
    C.-C. Yen, 
    Algorithmic aspects of perfect domination,
    Ph.D. Thesis, Institute of Information Science, National Tsing Hua University, Taiwan 1992.

\bibitem{YenLee1996}
    C.-C. Yen, R.C.T. Lee,
    The weighted perfect domination problem and its variants, 
{\sl Discrete Applied Math.} 66 (1996) 147-160.

\end{thebibliography}
\end{document}